\documentclass[11pt,aps,pra,notitlepage,tightenlines,nofootinbib,superscriptaddress]{revtex4-2}

\usepackage{newpxtext,newpxmath}

\let\coloneqq\relax
\let\eqqcolon\relax

\usepackage[latin1]{inputenc}
\usepackage{amsthm}
\usepackage{amssymb}
\usepackage{amsmath}
\usepackage{bbold}
\usepackage{bbm}
\usepackage[pdftex, backref=page]{hyperref}
\usepackage{braket}
\usepackage{dsfont}
\usepackage{mathdots}
\usepackage{mathtools}
\usepackage{enumerate}
\usepackage[shortlabels]{enumitem}
\usepackage{csquotes}
\usepackage{stmaryrd}
\usepackage[cal=boondox]{mathalfa}
\usepackage{graphicx}
\usepackage{stackengine}
\usepackage{scalerel}
\usepackage{tensor}       
\usepackage{array}
\usepackage{makecell}
\newcolumntype{x}[1]{>{\centering\arraybackslash}p{#1}}
\usepackage{tikz}
\usepackage{pgfplots}
\usetikzlibrary{shapes.geometric, shapes.misc, positioning, arrows, arrows.meta, decorations.pathreplacing, decorations.pathmorphing, patterns, angles, quotes, calc}
\usepackage{booktabs}
\usepackage{xfrac}
\usepackage{siunitx}
\usepackage{centernot}
\usepackage{comment}
\usepackage{chngcntr}
\usepackage{caption}
\usepackage{subcaption}

\newtheorem{thm}{Theorem}
\newtheorem*{thm*}{Theorem}
\newtheorem{prop}[thm]{Proposition}
\newtheorem*{prop*}{Proposition}
\newtheorem{lemma}[thm]{Lemma}
\newtheorem*{lemma*}{Lemma}
\newtheorem{cor}[thm]{Corollary}
\newtheorem*{cor*}{Corollary}
\newtheorem{cj}[thm]{Conjecture}
\newtheorem*{cj*}{Conjecture}

\newtheorem*{Def*}{Definition}

\newtheorem*{question*}{Question}

\newtheorem*{problem*}{Problem}

\makeatletter
\def\thmhead@plain#1#2#3{%
  \thmname{#1}\thmnumber{\@ifnotempty{#1}{ }\@upn{#2}}%
  \thmnote{ {\the\thm@notefont#3}}}
\let\thmhead\thmhead@plain
\makeatother

\theoremstyle{definition}
\newtheorem{rem}[thm]{Remark}
\newtheorem*{note}{Note}

\newcommand{\bb}{\begin{equation}\begin{aligned}\hspace{0pt}}
\newcommand{\bbb}{\begin{equation*}\begin{aligned}}
\newcommand{\ee}{\end{aligned}\end{equation}}
\newcommand{\eee}{\end{aligned}\end{equation*}}
\newcommand*{\coloneqq}{\mathrel{\vcenter{\baselineskip0.5ex \lineskiplimit0pt \hbox{\scriptsize.}\hbox{\scriptsize.}}} =}
\newcommand*{\eqqcolon}{= \mathrel{\vcenter{\baselineskip0.5ex \lineskiplimit0pt \hbox{\scriptsize.}\hbox{\scriptsize.}}}}

\newcommand\ceil[1]{\left\lceil#1\right\rceil}
\newcommand{\eqt}[1]{\stackrel{\mathclap{\scriptsize \mbox{#1}}}{=}}
\newcommand{\leqt}[1]{\stackrel{\mathclap{\scriptsize \mbox{#1}}}{\leq}}

\newcommand{\ketbra}[1]{\ket{#1}\!\!\bra{#1}}

\newcommand{\e}{\varepsilon}
\renewcommand{\epsilon}{\varepsilon}

\newcommand{\dd}{\mathrm{d}}

\newcommand{\id}{\mathds{1}}
\newcommand{\R}{\mathds{R}}
\newcommand{\N}{\mathds{N}}

\newcommand{\cptp}{\mathrm{CPTP}}

\newcommand{\locc}{\mathrm{LOCC}}

\newcommand{\SEP}{\pazocal{S}}

\newcommand{\sepp}{\mathrm{NE}}

\DeclareMathOperator{\Tr}{Tr}

\DeclareMathOperator{\cl}{cl}
\DeclareMathOperator{\co}{conv}

\DeclareMathAlphabet{\pazocal}{OMS}{zplm}{m}{n}

\newcommand{\HH}{\pazocal{H}}

\newcommand{\NN}{\pazocal{N}}

\newcommand{\MM}{\pazocal{M}}
\newcommand{\D}{\pazocal{D}}

\newcommand{\XX}{\pazocal{X}}

\newcommand{\FF}{\pazocal{F}}

\newcommand{\lsmatrix}{\left(\begin{smallmatrix}}
\newcommand{\rsmatrix}{\end{smallmatrix}\right)}

\newcommand{\deff}[1]{\textbf{\emph{#1}}}

\stackMath
\newcommand\xxrightarrow[2][]{\mathrel{%
  \setbox2=\hbox{\stackon{\scriptstyle#1}{\scriptstyle#2}}%
  \stackunder[5pt]{%
    \xrightarrow{\makebox[\dimexpr\wd2\relax]{$\scriptstyle#2$}}%
  }{%
   \scriptstyle#1\,%
  }%
}}

\newcommand{\ctends}[3]{\xxrightarrow[\raisebox{#3}{$\scriptstyle #2$}]{\raisebox{-0.7pt}{$\scriptstyle #1$}}}

\stackMath

\makeatletter
\newcommand*\rel@kern[1]{\kern#1\dimexpr\macc@kerna}
\newcommand*\widebar[1]{%
  \begingroup
  \def\mathaccent##1##2{%
    \rel@kern{0.8}%
    \overline{\rel@kern{-0.8}\macc@nucleus\rel@kern{0.2}}%
    \rel@kern{-0.2}%
  }%
  \macc@depth\@ne
  \let\math@bgroup\@empty \let\math@egroup\macc@set@skewchar
  \mathsurround\z@ \frozen@everymath{\mathgroup\macc@group\relax}%
  \macc@set@skewchar\relax
  \let\mathaccentV\macc@nested@a
  \macc@nested@a\relax111{#1}%
  \endgroup
}

\counterwithin*{equation}{part}
\counterwithin*{thm}{part}
\counterwithin*{figure}{part}

\tikzset{meter/.append style={draw, inner sep=10, rectangle, font=\vphantom{A}, minimum width=30, line width=.8, path picture={\draw[black] ([shift={(.1,.3)}]path picture bounding box.south west) to[bend left=50] ([shift={(-.1,.3)}]path picture bounding box.south east);\draw[black,-latex] ([shift={(0,.1)}]path picture bounding box.south) -- ([shift={(.3,-.1)}]path picture bounding box.north);}}}
\tikzset{roundnode/.append style={circle, draw=black, fill=gray!20, thick, minimum size=10mm}}
\tikzset{squarenode/.style={rectangle, draw=black, fill=none, thick, minimum size=10mm}}

\definecolor{Blues5seq1}{RGB}{239,243,255}
\definecolor{Blues5seq2}{RGB}{189,215,231}
\definecolor{Blues5seq3}{RGB}{107,174,214}
\definecolor{Blues5seq4}{RGB}{49,130,189}
\definecolor{Blues5seq5}{RGB}{8,81,156}

\definecolor{Greens5seq1}{RGB}{237,248,233}
\definecolor{Greens5seq2}{RGB}{186,228,179}
\definecolor{Greens5seq3}{RGB}{116,196,118}
\definecolor{Greens5seq4}{RGB}{49,163,84}
\definecolor{Greens5seq5}{RGB}{0,109,44}

\definecolor{Reds5seq1}{RGB}{254,229,217}
\definecolor{Reds5seq2}{RGB}{252,174,145}
\definecolor{Reds5seq3}{RGB}{251,106,74}
\definecolor{Reds5seq4}{RGB}{222,45,38}
\definecolor{Reds5seq5}{RGB}{165,15,21}

\allowdisplaybreaks

\newcommand{\LL}{\pazocal{L}}
\newcommand{\etilde}{\tilde{E}}

\begin{document}

\title{Continuity of entropies via integral representations}

\author{Mario Berta}
\affiliation{Institute for Quantum Information,
  RWTH Aachen University, Germany}

\author{Ludovico Lami}
\affiliation{QuSoft, Science Park 123, 1098 XG Amsterdam, The Netherlands}
\affiliation{Korteweg--de Vries Institute for Mathematics, University of Amsterdam, Science Park 105-107, 1098 XG Amsterdam, The Netherlands}
\affiliation{Institute for Theoretical Physics, University of Amsterdam, Science Park 904, 1098 XH Amsterdam, The Netherlands}

\author{Marco Tomamichel}
\affiliation{Center for Quantum Technologies, National University of Singapore}
\affiliation{Department of Electrical and Computer Engineering, College of Design and Engineering, National University of Singapore}

\begin{abstract}
We show that Frenkel's integral representation of the quantum relative entropy provides a natural framework to derive continuity bounds for quantum information measures. Our main general result is a dimension-independent semi-continuity relation for the quantum relative entropy with respect to the first argument. Using it, we obtain a number of results: (1)~a tight continuity relation for the conditional entropy in the case where the two states have equal marginals on the conditioning system, resolving a conjecture by Wilde in this special case; (2)~a stronger version of the Fannes--Audenaert inequality on quantum entropy; (3)~better estimates on the quantum capacity of approximately degradable channels; (4)~an improved continuity relation for the entanglement cost; (5)~general upper bounds on asymptotic transformation rates in infinite-dimensional entanglement theory; and (6)~a proof of a conjecture due to Christandl, Ferrara, and Lancien on the continuity of `filtered' relative entropy distances. 
\end{abstract}

\maketitle

\tableofcontents


\section{Introduction}

In a series of recent works, a new integral formula for Umegaki's quantum relative entropy was discovered. Originally found by Frenkel~\cite{Frenkel2023} and later refined by Jen\v{c}ova~\cite{Jencova2023} and Hirche and Tomamichel~\cite{Hirche2023}, it can be written as
\bb
D(\rho\|\sigma) = (\log e) \int_1^\infty \dd\gamma \left(\frac{1}{\gamma}\, E_\gamma(\rho\|\sigma) + \frac{1}{\gamma^2}\, E_\gamma(\sigma\|\rho)\right) .
\label{eq:integral_formula}
\ee
Here, Umegaki's quantum \deff{relative entropy}~\cite{Umegaki1962} of a quantum state $\rho$ with regards to another quantum state $\sigma$ is given as
\bb
D(\rho\|\sigma) \coloneqq \Tr \rho \left(\log \rho - \log \sigma\right) ,
\ee
assuming that the support $\rho$ is contained in the support $\sigma$ as otherwise the quantity is not finite. Moreover, for any given $\gamma \geq 1$, the corresponding \deff{hockey-stick divergence} is defined by
\bb
E_\gamma(\rho\|\sigma) \coloneqq \Tr \left(\rho - \gamma\sigma\right)_+\, ,
\ee
where $X_+ \coloneqq \sum_i \max\{x_i,0\} \ketbra{i}$ is the positive part of a Hermitian operator $X$ with spectral decomposition $X = \sum_i x_i \ketbra{i}$.

The quantum relative entropy can be seen as a parent quantity for quantum entropies, as other information measures such as the von Neumann entropy, the conditional entropy, and the mutual information can all be expressed in terms of it. The operational importance of the relative entropy, in turn, stems from the quantum Stein's lemma due to Hiai and Petz~\cite{Hiai1991}, which endows it with an operational interpretation in the context of quantum hypothesis testing. The above expression of the relative entropy in terms of hockey-stick divergences is particularly fruitful because the latter quantities, thanks to their operational significance in binary hypothesis testing, exhibit many useful mathematical properties and are thus comparatively easy to manipulate. 

Integral representations of quantum relative entropies have already been employed to derive continuity bounds (see~\cite{Audenaert2005} for an early example and~\cite[Section~5.2]{Hirche2023} for the most recent instance). 
One of our main goals is to derive fully quantum extensions of the continuity inequality for the Shannon entropy. Its tight version, due to Csisz\'ar, states that for two probability distributions $P_X,Q_X$ on a finite alphabet $\XX$ with $\frac{1}{2}\|P_X-Q_X\|_1\leq \e \leq 1-\frac{1}{|\XX|}$ in variational distance, one has that
\bb
\left|H(X)_P-H(X)_Q\right|\leq \e \log\left(|\XX| - 1\right) +h_2(\e)\,,
\label{eq:Csiszar}
\ee
where $H(X)_P\coloneqq-\sum_{x\in \XX} P(x)\log P(x)$ denotes the Shannon entropy, and 
\bb
  h_2(p) \coloneqq p \log \frac{1}{p} + (1-p) \log \frac{1}{1-p}
\ee
is the \deff{binary entropy function}. The above inequality admits a quantum generalisation due to Fannes~\cite{Fannes1973}, Audenaert~\cite{Audenaert2007}, and Petz~\cite[Theorem~3.8]{PETZ-QINFO}:
\bb
\left| S(\rho) - S(\sigma)\right| \leq \e \log(d-1) + h_2(\e)\, ,
\label{eq:Fannes_Audenaert_Petz}
\ee
where $S(\rho)\coloneqq - \Tr \rho \log \rho$, called the \deff{von Neumann entropy}, is the extension of the Shannon entropy to density operators $\rho$ on Hilbert spaces of finite dimension $d$, and $\e$, which is assumed to satisfy $\frac12 \|\rho-\sigma\|_1 \leq \e \leq 1 - \frac{1}{d}$, is a bound on the \deff{trace distance} between the density operators $\rho$ and $\sigma$, where $\|\rho-\sigma\|_1 \coloneqq \Tr |\rho-\sigma|$. Yet another generalisation of~\eqref{eq:Csiszar} is to conditional entropies: for any two probability distributions $P_{XY},Q_{XY}$ on a product alphabet $\XX\times \pazocal{Y}$, if $\frac{1}{2}\|P_{XY} - Q_{XY} \|_1\leq \e \leq 1-\frac{1}{|\XX|}$ then Alhejji and Smith proved that~\cite{Alhejji2019}
\bb
\left|H(X|Y)_P - H(X|Y)_Q\right|\leq \e \log\left(|\XX| - 1\right) +h_2(\e)\, ,
\label{eq:Alhejji_Smith}
\ee
where $H(X|Y) \coloneqq H(XY) - H(Y)$ is the conditional entropy. The above inequality was then extended by Wilde to the quantum-classical case where $X$, the conditioned system, is promoted to a quantum system~\cite{Mark2020}. Wilde also conjectured that a similar continuity relation should hold also for the fully quantum case~\cite[Eq.~(58)]{Mark2020}, namely that
\bb
\left| H(A|B)_\rho - H(A|B)_\sigma \right| \leqt{?} \e \log\big(|A|^2 - 1 \big) + h_2(\e)\, ,
\label{eq:Wilde_conjecture}
\ee
where $H(A|B)_\rho \coloneqq S(\rho_{AB}) - S(\rho_B)$ is the conditional entropy, $|A|$ is the dimension of the quantum system $A$, and $\frac12 \left\|\rho_{AB} - \sigma_{AB}\right\|_1 \leq \e \leq 1 - \frac{1}{|A|^2}$ is an estimate on the trace distance between the two states $\rho_{AB}$ and $\sigma_{AB}$.

Our intuition to unify all these different bounds and move towards a proof of~\eqref{eq:Wilde_conjecture} is to lift the problem from quantum entropies to the corresponding parent quantity, the relative entropy. Our main general result is Theorem~\ref{semicontinuity_relent_thm} below, which presents a semi-continuity bound of the form
\bb
D(\rho\|\omega) - D(\sigma\|\omega) \leq \e \log(M-1) + h_2(\e)\, ,
\label{semicontinuity_relent_informal}
\ee
where $\frac12 \|\rho-\sigma\|_1\leq \e \leq 1 - \frac{1}{M}$, and $M$ is any number such that the operator inequality $\rho \leq M \omega$, meaning that $M\omega - \rho$ is positive semi-definite, holds. Note that the left-hand side of~\eqref{semicontinuity_relent_informal} does not contain an absolute value. This feature, which makes ours a \emph{semi}-continuity bound rather than a plain continuity bound like~\eqref{eq:Csiszar}, is rather fundamental, as the expression on the left-hand side can happen to diverge to $-\infty$. Eq.~\eqref{semicontinuity_relent_informal} is in fact tight, in the sense that for all $M\geq 1$ and $\e\in \big[0,1-\tfrac1M\big]$ one can find a triple of states $\rho,\sigma,\omega$ obeying the above conditions and saturating~\eqref{semicontinuity_relent_informal}. Our proof of~\eqref{semicontinuity_relent_informal} uses the integral representation in~\eqref{eq:integral_formula} together with the properties of hockey-stick divergences. The conceptual novelty of our approach, therefore, is that it works directly at the level of quantum entropies, without first operating a reduction to the classical case, like all previous approaches to prove the Fannes--Audenaert inequality~\cite{Fannes1973, Audenaert2007, PETZ-QINFO}.

The formal resemblance between~\eqref{eq:Csiszar} and~\eqref{eq:Fannes_Audenaert_Petz}, on one side, and~\eqref{semicontinuity_relent_informal}, on the other, is clear. Indeed, as we will show, the latter inequality constitutes a strict generalisation of the former two. In fact, in Section~\ref{sec:improved_FAP} we show how to use our~\eqref{semicontinuity_relent_informal} to obtain an improved version of the Fannes--Audenaert inequality~\eqref{eq:Fannes_Audenaert_Petz} (Corollary~\ref{improved_Fannes_Audenaert_cor}). 

As an immediate consequence of~\eqref{semicontinuity_relent_informal}, in Section~\ref{sec:conditional_entropy} we also prove Wilde's conjecture~\eqref{eq:Wilde_conjecture} in the special case where $\rho_{AB}$ and $\sigma_{AB}$ have equal marginals on the conditioning system, in formula $\rho_B = \sigma_B$. 
Our inequality takes the form
\bb
\left| H(A|B)_\rho - H(A|B)_\sigma \right| \leq \e \log\left(|A|\min\{|A|,|B|\} - 1\right) + h_2(\e) \qquad \qquad (\rho_B = \sigma_B)\, ,
\label{continuity_conditional_entropy_informal}
\ee
where $\frac12 \left\|\rho_{AB} - \sigma_{AB}\right\|_1\leq \e \leq 1 - \left(|A| \min\{|A|,|B|\}\right)^{-1}$; it generalises the Fannes--Audenaert inequality~\eqref{eq:Fannes_Audenaert_Petz}, and it directly implies~\eqref{eq:Wilde_conjecture} in the case of equal marginals. Note that the case where $|B|\leq |A|$ is less interesting, as it follows immediately from~\eqref{eq:Fannes_Audenaert_Petz} applied to the global system $AB$. The case where the dimension of $B$ is larger than that of $A$, and perhaps even infinite, is, on the contrary, novel.

Remarkably, in proving~\eqref{continuity_conditional_entropy_informal} we cannot afford to restrict ourselves to the classical case, as one typically does to prove~\eqref{eq:Fannes_Audenaert_Petz}, because the conditional entropy, unlike the von Neumann entropy, is not invariant under (global) unitaries. It is here that our fully quantum methods grant us a decisive advantage.

While our solution of Wilde's conjecture covers only the case of equal marginals, it is precisely this case that turns out to have a plethora of fruitful applications throughout quantum information theory. Using~\eqref{continuity_conditional_entropy_informal}, we are able to derive continuity bounds on several fundamental quantities that are tighter than anything found in the prior literature. First, in Section~\ref{sec:continuity_entanglement_cost} we tackle the entanglement cost, a key entanglement measure that quantifies how much entanglement is needed to prepare a given quantum state in the asymptotic limit. Then, in Section~\ref{sec:approximate_degradability}, we employ once again~\eqref{continuity_conditional_entropy_informal} to derive the tightest available bounds on the magnitude of additivity violations in the quantum capacity of approximately degradable channels~\cite{Sutter-VV}. We continue by applying our general result~\eqref{semicontinuity_relent_informal} in a somewhat indirect way to resolve a conjecture by Christandl, Ferrara, and Lancien~\cite[Conjecture~7]{random-private} on a special form of continuity of certain relative-entropy-based resource quantifiers. Finally, in Section~\ref{sec:transformation_rates} we show how to employ the dimension-independent inequality~\eqref{semicontinuity_relent_informal} to constrain transformation rates in infinite-dimensional entanglement theory, by-passing known technical hurdles related to the failure of asymptotic continuity in infinite-dimensional systems, a phenomenon referred to as the `asymptotic continuity catastrophe' in~\cite{nonclassicality}.


\section{A semi-continuity relation for the relative entropy}

We now present our main technical result, which employs the integral formula in~\eqref{eq:integral_formula} to derive a new semi-continuity relation for the relative entropy with respect to the first argument, when the second one is fixed.

\begin{thm} \label{semicontinuity_relent_thm}
Let $\rho,\sigma,\omega$ be quantum states on the same separable Hilbert space such that
\bb
\frac12 \left\|\rho - \sigma\right\|_1 \leq \e\, ,\qquad \exp\left[D_{\max}(\rho\|\omega)\right] := \inf\{ m :\ \rho\leq m \omega \} \leq M\,
\label{semicontinuity_relent_thm_conditions_states}
\ee
for some $\e \in [0,1]$ and $M \geq 1$.
Then, we have that
\bb
D(\rho\|\omega) - D(\sigma\|\omega) \leq \left\{ \begin{array}{ll} \e \log \left(M-1\right) + h_2(\e) & \quad \text{ if $\e < 1-\frac1M$,} \\[2ex] \log M & \quad \text{ if $\e \geq 1-\frac1M$.} \end{array} \right.
\label{semicontinuity_relent}
\ee
Moreover, for every such $\e$ and $M$, there exists a triple of states $\rho, \sigma$ and $\omega$ obeying~\eqref{semicontinuity_relent_thm_conditions_states} for which the inequality~\eqref{semicontinuity_relent} is tight.
\end{thm}

\begin{rem} \label{simplified_semicontinuity_relent_rem}
It is possible to give a slightly simplified bound that is looser than the right-hand side of~\eqref{semicontinuity_relent} but does not involve a function defined piece-wise. Namely, it holds that
\bb
D(\rho\|\omega) - D(\sigma\|\omega) \leq \e \log M + h_2(\e)\, ,
\label{simplified_semicontinuity_relent}
\ee
simply because the right-hand side of the above equation can be shown to be never smaller than that of~\eqref{semicontinuity_relent}.
\end{rem}

\begin{proof}[Proof of Theorem~\ref{semicontinuity_relent_thm}]
If $\e \geq 1-1/M$ there is nothing to prove, as it is clear that
\bb
D(\rho\|\omega) - D(\sigma\|\omega) \leq D(\rho\|\omega) \leq D_{\max}(\rho\|\omega) \leq \log M\, ;
\ee
therefore, let us assume that $\e < 1-1/M$ and thus also $M > 1$.
We start by writing
\bb
\frac{1}{\log e} \left( D(\rho\|\omega) - D(\sigma\|\omega) \right) &= \int_1^\infty \frac{\dd\gamma}{\gamma} \left(E_\gamma(\rho\|\omega) - E_\gamma(\sigma\|\omega) \right) + \int_1^\infty \frac{\dd\gamma}{\gamma^2} \left(E_\gamma(\omega\|\rho) - E_\gamma(\omega\|\sigma) \right) .
\label{splitting_two_integrals}
\ee
We will now bound the two above integrals separately. First, it holds that
\bb
&\int_1^\infty \frac{\dd\gamma}{\gamma} \left(E_\gamma(\rho\|\omega) - E_\gamma(\sigma\|\omega) \right) \\
&\qquad \eqt{(i)} \int_1^{(1-\e)M} \frac{\dd\gamma}{\gamma} \left(E_\gamma(\rho\|\omega) - E_\gamma(\sigma\|\omega) \right) + \int_{(1-\e)M}^M \frac{\dd\gamma}{\gamma} \left(E_\gamma(\rho\|\omega) - E_\gamma(\sigma\|\omega) \right) \\
&\qquad\qquad + \int_M^\infty \frac{\dd\gamma}{\gamma} \left(E_\gamma(\rho\|\omega) - E_\gamma(\sigma\|\omega) \right) \\
&\qquad \leqt{(ii)} \int_1^{(1-\e)M}\frac{\dd\gamma}{\gamma}\ \e + \int_{(1-\e)M}^M \frac{\dd\gamma}{\gamma} \left( 1 - \frac{\gamma}{M} \right) \\
&\qquad = \e \ln \left((1-\e)M\right) + \ln \frac{1}{1-\e} - \e \\
&\qquad = - (1-\e) \ln (1-\e) + \e \ln M - \e\, .
\label{first_integral_bound}
\ee
Here, in~(i) we split the integral into three parts, noting that $(1-\e)M\geq 1$. The justification of~(ii) is a bit more involved, as we used three different upper bounds for the integrand in each of the three regions: namely, for $1 \leq \gamma \leq (1-\e)M$ we have that
\bb
E_\gamma(\rho\|\omega) = \Tr \left(\rho - \gamma \omega\right)_+ \leq \Tr \left(\rho - \sigma\right)_+ + \Tr\left(\sigma - \gamma \omega\right)_+ \leq \e + E_\gamma(\sigma\|\omega)\, ;
\ee
in the third region, i.e.\ for $\gamma\geq M$, it holds that
\bb
E_\gamma(\rho\|\omega) - E_\gamma(\sigma\|\omega) \leq E_\gamma(\rho\|\omega) \leq E_M(\rho\|\omega) = \Tr \left(\rho - M \omega\right)_+ = 0\, ,
\ee
where the last equality follows because $\rho\leq M\omega$ by definition; finally, in the second region, i.e.\ for $(1-\e)M \leq \gamma \leq M$, we write
\bb
\gamma = pM + (1-p) (1-\e)M\, ,\qquad p \coloneqq \frac{\gamma - (1-\e)M}{\e M} ,
\ee
and then
\bb
E_\gamma(\rho\|\omega) &\leq p\, E_M(\rho\|\omega) + (1-p)\, E_{(1-\e)M}(\rho\|\omega) \\
&= (1-p)\, E_{(1-\e)M}(\rho\|\omega) \\
&= (1-p) \Tr \left( \rho - (1-\e)M \omega \right)_+ \\
&\leq (1-p) \Tr \left( \rho - (1-\e) \rho \right)_+ \\
&= (1-p)\, \e \\
&= 1 - \frac{\gamma}{M}\,,
\ee
where the first inequality is by convexity of $\gamma \mapsto E_\gamma(\rho\|\omega)$, in the second line we observed again that $E_M(\rho\|\omega) = 0$, and the inequality in the fourth line holds due to the operator inequality $M\omega \geq \rho$, together with the fact that $\Tr X_+$ is monotonically non-decreasing in $X$ with respect to the positive semi-definite (L\"owner) ordering. This completes the justification of~\eqref{first_integral_bound}.

We now move on to the second integral in~\eqref{splitting_two_integrals}. It holds that
\bb
&\int_1^\infty \frac{\dd\gamma}{\gamma^2} \left(E_\gamma(\omega\|\rho) - E_\gamma(\omega\|\sigma) \right) \\
&\qquad \eqt{(iii)} \int_1^{\frac{M-1}{M\e}} \frac{\dd\gamma}{\gamma^2} \left(E_\gamma(\omega\|\rho) - E_\gamma(\omega\|\sigma) \right) + \int_{\frac{M-1}{M\e}}^\infty \frac{\dd\gamma}{\gamma^2} \left(E_\gamma(\omega\|\rho) - E_\gamma(\omega\|\sigma) \right) \\
&\qquad \leqt{(iv)} \int_1^{\frac{M-1}{M\e}} \frac{\dd\gamma}{\gamma^2}\ \gamma \e + \int_{\frac{M-1}{M\e}}^\infty \frac{\dd\gamma}{\gamma^2} \left(1-\frac1M \right) \\
&\qquad = \e \ln \frac{M-1}{M\e} + \left(1-\frac1M\right) \frac{M\e}{M-1} \\
&\qquad = \e \ln (M-1) - \e \ln M - \e \ln \e + \e\, .
\label{second_integral_bound}
\ee
Here: in~(iii) we split the integral, noting that $\frac{M-1}{M\e}\geq 1$ by assumption; in~(iv) we observed that
\bb
E_\gamma(\omega\|\rho) = \Tr \left(\omega - \gamma \rho\right)_+ \leq \Tr \left(\omega - \gamma \sigma\right)_+ + \gamma \Tr (\sigma - \rho)_+ \leq E_\gamma(\omega\|\sigma) + \gamma \e\, ,
\ee
which gives the upper bound on the first term, and moreover
\bb
E_\gamma(\omega\|\rho) - E_\gamma(\omega\|\sigma) \leq E_\gamma(\omega\|\rho) \leq E_{1/M}(\omega\|\rho) = \Tr \left( \omega - \frac1M\, \rho \right)_+ = 1 - \frac1M\, ,
\ee
because rather obviously
\bb
\gamma \geq \frac{M-1}{M\e} \geq 1 \geq \frac1M\, ,
\ee
which gives the upper bound on the second term. This completes the justification of~\eqref{second_integral_bound}. Putting~\eqref{first_integral_bound} and~\eqref{second_integral_bound} together finally proves the claim.

\smallskip

To show optimality, consider (commuting) qubit states
\bb
    \rho = \ketbra{0}, \quad \sigma = (1-\e) \ketbra{0} + \e \ketbra{1}  , \quad
    \textnormal{and} \quad 
    \omega = \frac{1}{M} \ketbra{0} + \left( 1 - \frac{1}{M} \right) \ketbra{1}\, ,
\ee
which satisfy $\frac12 \| \rho - \sigma \|_1 = \e$ and $D_{\max}(\rho\|\omega) = \log M$, as well as
\bb
D(\rho\|\omega) - D(\sigma\|\omega) = \e \log (M-1) + h_2(\e) \eqqcolon f(M, \e)
\ee
We note that $f$ is monotone in $M$ but achieves a maximum of $\log M$ at $\e = 1 - \frac{1}{M}$ as a function of $\e$. If $\e\leq 1-\frac1M$, the inequality in~\eqref{semicontinuity_relent} is clearly saturated. If $\e > 1 - \frac1M$, instead, we can consider the above three states for $\e \mapsto 1 - \frac1M$. Since we only require that $\frac12 \| \rho - \sigma \|_1 \leq \e$, this once again proves the claim.
\end{proof}


\section{Applications}

\subsection{Improved Fannes--Audenaert inequality} \label{sec:improved_FAP}

Theorem~\ref{semicontinuity_relent_thm} evaluated for $\omega=\pi \coloneqq \id/d$ immediately implies a slightly strengthened version of the Fannes--Audenaert inequality in~\eqref{eq:Fannes_Audenaert_Petz}. We recall that this inequality states that given two states $\rho,\sigma$ at trace distance $\e \coloneqq \frac12 \left\|\rho-\sigma\right\|_1$, their entropies obey the inequality $\left|S(\rho) - S(\sigma)\right| \leq \e \log(d-1) + h_2(\e)$, where $d$ is the dimension of the underlying Hilbert space. The original, looser form of this continuity relation was found by Fannes~\cite{Fannes1973}, while the tight version was established by Audenaert~\cite{Audenaert2007} and Petz~\cite[Theorem~3.8]{PETZ-QINFO}. Our bound yields an improvement for mixed states.

\begin{cor} \label{improved_Fannes_Audenaert_cor}
Let $\rho,\sigma$ be any two states at trace distance $\frac12 \|\rho - \sigma\|_1 \leq \e \leq 1 - \frac{1}{d\,\lambda_{\max}(\sigma)}$ on a $d$-dimensional quantum system. Here, $\lambda_{\max}(\sigma)$ denotes the maximal eigenvalue of $\sigma$. Then, we have that
\bb
S(\rho) - S(\sigma) \leq \e \log\left( d\,\lambda_{\max}(\sigma) - 1 \right) + h_2(\e)\, .
\label{improved_Fannes_Audenaert}
\ee
\end{cor}

\begin{rem} \label{restriction_range_rem}
The restriction on the range of $\e$ in the statement of Corollary~\ref{improved_Fannes_Audenaert_cor} is motivated by considerations that are similar to those found at the end of the proof of Theorem~\ref{semicontinuity_relent_thm}. Namely, for $\e = 1 - \frac{1}{d\,\lambda_{\max}(\sigma)}$ the right-hand side of~\eqref{improved_Fannes_Audenaert} achieves its maximum value, equal to $\log\big( d\lambda_{\max}(\sigma)\big)$. That the left-hand side of~\eqref{improved_Fannes_Audenaert} must be upper bounded by this number is however easy to show, so for larger values of $\e$ the statement of Corollary~\ref{improved_Fannes_Audenaert_cor} would be somewhat empty. 

Another approach to avoid restricting the range of $\e$ is to define $\e$ as the exact trace distance $\frac12 \|\rho-\sigma\|_1$, instead of taking it to be a generic upper bound on it. While this is the approach followed by some authors, it is not what is most useful in applications, where one often has an estimate of the trace distance but not its exact value. Similar considerations will apply to the restrictions on the range of $\e$ in, e.g., the forthcoming Theorem~\ref{equal_marginals_thm} and Proposition~\ref{prop:E-cost}.
\end{rem}

We stress that our proof, in contrast to previous arguments, does not go via a reduction to the commutative case (e.g.\ via dephasing) but instead the integral representation method applies directly in the quantum setting.


\subsection{Improved continuity of conditional entropy}
\label{sec:conditional_entropy}

Our Theorem~\ref{semicontinuity_relent_thm} can be applied to the question about the tight continuity of the quantum conditional entropy independent of the dimension of the conditioning system~\cite{squashed,Alicki-Fannes,tightuniform}. Our tightest bound features the mixed state Schmidt number~\cite{Terhal00}
\bb
\mathrm{SN}(\rho_{AB})\coloneqq\inf_{\rho=\sum_ip_i\psi_i} \max_i\, \mathrm{SN}\big(\ket{\psi_i}_{AB}\big)\,,
\ee
where the minimisation is over pure state convex decompositions $\rho_{AB} = \sum_ip_i\ketbra{\psi_i}_{AB}$, and $\mathrm{SN}\big(\ket{\psi_i}_{AB}\big)$ denotes the regular pure state Schmidt number, i.e.\ the rank of the reduced state $\Tr_B \ketbra{\psi_i}_{AB}$.

\begin{thm}
\label{equal_marginals_thm}
Let $AB$ be a bipartite quantum system where $A$ has finite dimension $|A|$. Let $\rho = \rho_{AB}$ and $\sigma = \sigma_{AB}$ be two states at trace distance $\frac12 \|\rho - \sigma\|_1 \leq \e$ and having the same marginal on $B$, i.e.\ such that
\bb
\rho_B = \sigma_B\, .
\ee
Then, assuming $\e\leq 1-\frac{1}{|A|\, \mathrm{SN}(\rho_{AB})}$, it holds that
\bb
H(A|B)_\rho - H(A|B)_\sigma\leq \e \log\big( |A|\cdot\mathrm{SN}(\rho_{AB}) - 1 \big) + h_2(\e)\,.
\ee
\end{thm}

It then immediately follows that
\bb
\label{eq:main-bound}
\left|H(A|B)_\rho - H(A|B)_\sigma\right| \leq \e \log\big(|A|\min\{|A|,|B|\} - 1\big) + h_2(\e)
\ee
for quantum states $\rho_{AB},\sigma_{AB}$ with $\rho_B = \sigma_B$ and $\e\leq 1 - \frac{1}{|A|\,\min\{|A|,|B|\}}$, whereas for separable states the bound is improved to
\bb
\left|H(A|B)_\rho - H(A|B)_\sigma\right|\leq\e \log\left( |A| - 1 \right) + h_2(\e)
\ee
for $\e\leq 1 - \frac{1}{|A|}$. This pleasantly reproduces the tight classical~\eqref{eq:Alhejji_Smith} from~\cite{Alhejji2019} and the corresponding quantum-classical bound~\cite{Mark2020} (that both equally hold even when $\rho_B\neq\sigma_B$).

\begin{proof}[Proof of Theorem~\ref{equal_marginals_thm}]
Set 
\bb
\omega = \omega_{AB} = \pi_A \otimes \rho_B = \pi_A \otimes \sigma_B\, .
\ee
For any positive integer $k\in \N^+$, the map $X \mapsto k \id \Tr X - X$ is known to be $k$-positive (see e.g.~\cite[Eq.~(3.11) and discussion thereafter]{WolfQC}). By taking $k=\mathrm{SN}(\rho_{AB})$ and applying the map to $\rho_{AB}$, we deduce the operator inequality
\bb
\rho_{AB} \leq \mathrm{SN}(\rho_{AB})\, \id_A \otimes \rho_B = |A|\cdot \mathrm{SN}(\rho_{AB})\, \pi_A\otimes \rho_B = |A|\cdot \mathrm{SN}(\rho_{AB})\, \omega_{AB}\, , 
\ee
which immediately implies that
\bb
D_{\max}\big(\rho_{AB}\,\big\|\, \omega_{AB}\big) \leq \log \big(|A|\cdot \mathrm{SN}(\rho_{AB})\big)\, .
\ee
We thus obtain that
\bb
H(A|B)_\sigma - H(A|B)_\rho &= D\big(\rho_{AB}\,\big\|\, \pi_A \otimes \rho_B\big) - D\big(\sigma_{AB}\,\big\|\, \pi_A \otimes \sigma_B\big) \\
&= D(\rho_{AB}\|\omega_{AB}) - D(\sigma_{AB}\|\omega_{AB}) \\
&\leq \e \log\big(|A|\cdot \mathrm{SN}(\rho_{AB}) - 1 \big) + h_2(\e)\, ,
\ee
where in the last step we applied Theorem~\ref{semicontinuity_relent_thm}.
\end{proof}

As mentioned in the introduction, it was then also conjectured by Wilde~\cite{Mark2020} that~\eqref{eq:main-bound} holds generally for finite-dimensional states:

\begin{cj}[{\cite{Mark2020}}]
\label{cj:wilde-conjecture}
For two finite-dimensional quantum states $\rho_{AB},\sigma_{AB}$ with trace distance $\frac12 \|\rho_{AB} - \sigma_{AB}\|_1 \leq \e\leq 1-\frac{1}{|A|^2}$, do we have that
\bb
\left|H(A|B)_\rho - H(A|B)_\sigma\right|\leqt{?} \e \log\big( |A|^2 - 1 \big) + h_2(\e)\,.
\ee
\end{cj}

This would indeed be stronger than the state-of-the-art Alicki--Fannes--Winter bound~\cite{Alicki-Fannes, Synak2006, mosonyi11, Winter2016}, which states for $\frac12 \|\rho_{AB} - \sigma_{AB}\|_1 \leq \e\leq1$ that
\bb
\label{eq:Alicki-Fannes}
\left|H(A|B)_\rho - H(A|B)_\sigma\right|\leq\e\log|A|^2+(1+\e)\cdot h_2\left(\frac{\e}{1+\e}\right)\,.
\ee
While we fail to prove Conjecture~\ref{cj:wilde-conjecture} in full generality, a multitude of comments around our Theorem~\ref{equal_marginals_thm} are noteworthy:
\begin{itemize}
  \item Eq.~\eqref{eq:main-bound} is exactly tight for $|A|=|B|$ in all dimensions, as choosing $\rho_{AB}=\Phi_{AB}$ maximally entangled and $\sigma_{AB}=(1-\varepsilon)\Phi_{AB}+\frac{\varepsilon}{|A|^2-1}\left(\id_{AB} -\Phi_{AB}\right)$ gives equality. This is in contrast to the not exactly tight~\eqref{eq:Alicki-Fannes}.
  
  \item For the general case, our Theorem~\ref{semicontinuity_relent_thm} implies that
  \bb
  H(A|B)_\sigma-H(A|B)_\rho\leq \e \log\big(|A|\min\{|A|,|B|\} - 1\big) + h_2(\e)+D(\sigma_B\|\rho_B)\,.
  \ee
  
  \item Employing the Fannes--Audenaert inequality on $AB$ for the special case $\rho_B=\sigma_B$ yields
  \bb
  \left|H(A|B)_\rho-H(A|B)_\sigma\right|\leq\e\log(|A||B|-1)+h_2(\e)\,.
  \ee
  
  \item Following~\cite[Proposition 3.5.3]{AlhejjiPhD}, one resolved quantum case is when $H(A|B)_\rho\geq H(A|B)_\sigma$ and $\sigma_{AB}$ is diagonal in a maximally entangled basis (which also implies that $|A|=|B|$). This, however, then also directly follows from Fannes--Audenaert inequality\,---\,as remarked in~\cite[Section 3.5]{AlhejjiPhD}.
  
  \item In the classical case, one does have the tight bound~\eqref{eq:Alhejji_Smith}, which is not directly implied by Theorem~\ref{semicontinuity_relent_thm}. The classical proof from~\cite{Alhejji2019} employs sophisticated pre-processing of the input distributions $P_{XY}$ and $Q_{XY}$, making use of the conditional majorisation monotonicity of the conditional entropy. This then reduces everything to the product case with a uniform $Y$-marginal. The proof method also extends to $A$ quantum, by means of conditional dephasing~\cite{Mark2020}.
  
  \item For the general quantum case, one might then try similar pre-processing ideas as well, making use of the monotonicity of the quantum conditional entropy under $A$-unital operations~\cite{Vempati22}, also known as quantum conditional majorisation~\cite{Brandsen21,Gour18}\,---\,which is also how the above example of $\sigma_{AB}$ diagonal in a maximally entangled basis works.
  
  \item Another approach is to employ divergence centers for the choice of $\omega$ in Theorem~\ref{semicontinuity_relent_thm}. For the quantum relative entropy it is defined as~\cite{PETZ-ENTROPY}
  \bb
  \omega^*(\rho,\sigma)\coloneqq\mathrm{argmin}\max \big\{D(\rho\|\omega),D(\sigma\|\omega)\big\}\,,
  \ee
  for which $D(\sigma_B\|\omega^*_B)-D(\rho_B\|\omega^*_B)=0$. However, it is not a priori clear how to tightly upper bound $D_{\max}(\rho_{AB}\|\pi_A\otimes\omega^*_B)$ for that choice. It is known that~\cite{PETZ-ENTROPY,tomamicheltan14}
  \bb
  \text{$\omega^*(\rho,\sigma)=\lambda\cdot\rho+(1-\lambda)\sigma$ for $\lambda\in[0,1]$}
  \ee
  and one might choose some general $\lambda$ with $\omega_B^\lambda(\rho_B,\sigma_B)$. Bounding
  \bb
  D_{\max}\big(\rho_{AB}\, \big\|\, \pi_A\otimes\omega_B^\lambda\big) \leq \log\left(\lambda^{-1} d^2\right)
  \ee
  and then further optimising the choice of $\lambda\in[0,1]$ against the terms
  \bb
  D\big(\sigma_B\, \big\|\, \omega_B^\lambda\big) - D\big(\rho_B\, \big\|\, \omega_B^\lambda\big)
  \ee
  also leads to some bounds. In fact, this approach leads to the conjectured bound as long as $H_{\min}(A|B)_\rho\geq-\log\big(\lambda|A|\big)$, solidifying the intuition that the most difficult case to handle is when $\rho_{AB}$ is highly entangled.
  
  \item The original proofs~\cite{tightuniform, Alicki-Fannes} on the continuity of the quantum conditional entropy instead work with the max-relative entropy divergence center~\cite{mosonyi11}
  \bb
  \omega^*_{\max}(\rho,\sigma) \coloneqq \mathrm{argmin}\max \big\{D_{\max}(\rho\|\omega),D_{\max}(\sigma\|\omega)\big\}\, ,
  \ee
  where the \deff{max-relative entropy} between two states $\tau$ and $\xi$ on the same separable Hilbert space is defined by~\cite{Datta08}
  \bb
    D_{\max}(\tau\|\xi) \coloneqq \inf\left\{ \lambda:\, \tau\leq 2^\lambda \xi \right\} .
    \label{max_relative_entropy}
  \ee
  The above state $\omega^*_{\max}(\rho,\sigma)$ takes the explicit form
  \bb
  \omega^*_{\max}(\rho,\sigma)=\frac{\rho+2[\rho-\sigma]_+}{1+\frac{1}{2}\|\rho-\sigma\|_1}
  =\frac{\sigma+2[\rho-\sigma]_-}{1+\frac{1}{2}\|\rho-\sigma\|_1}\,.
  \ee
  This choice also leads to bounds in our formalism, which, however, do not seem to be exactly tight either.
\end{itemize}

We remark that Theorem~\ref{equal_marginals_thm} can be generalised with a suitable application of the chain rule, showing that only the marginals that differ contribute to the continuity bound. Namely, for $ABC$ a tripartite quantum system with $A$ of finite dimension $|A|$ and $\rho = \rho_{ABC}$ and $\sigma = \sigma_{ABC}$ two states at trace distance $\frac12 \|\rho-\sigma\|_1 \leq \e \leq 1 - \frac{1}{|A|^2}$ and satisfying $\rho_{BC} = \sigma_{BC}$, we have that
\bb
\left| H(AB|C)_\rho - H(AB|C)_\sigma \right| \leq \e \log \big( |A|^2 - 1\big) + h_2(\e)\, .
\ee
This follows from $H(AB|C) = H(A|BC) + H(B|C)$ together with the fact that the latter entropy has the same value on $\rho_{BC}$ and $\sigma_{BC}$.

Further, this then also implies bounds for entropic quantum channel expressions (cf.~\cite{Leung2009,Shirokov17}). Those often involve so-called regularisations~\cite{MARK}, which make the derivation of the corresponding continuity statements slightly more involved. As an example, we state an improved continuity bound for the output entropy of $n$-fold tensor product channels in terms of the diamond norm, which, for quantum channels $\MM_{A\to B}$ and $\NN_{A\to B}$, is defined as
\bb
\big\|\MM_{A\to B}-\NN_{A\to B}\big\|_{\diamond} \coloneqq \sup_{\ket{\Psi}_{AA'}} \left\|\Big(\big(\MM_{A\to B}-\NN_{A\to B}\big)\otimes I_{A'} \Big) \Big(\ketbra{\Psi}_{AA'}\Big)\right\|_1\,,
\label{diamond_norm}
\ee
where the supremum is over all ancillary systems $A'$ and all pure states $\ket{\Psi}_{AA'}$ on $A'$.

\begin{lemma}\label{lem:output-entropy}
Let $\MM_{A\to B}$ and $\NN_{A\to B}$ be quantum channels, with $B$ finite dimensional and $\frac12\big\|\MM_{A\to B}-\NN_{A\to B}\big\|_{\diamond}\leq\e\leq1-\frac{1}{|B|^2}$. Then, for all reference systems $A'$ and states $\rho_{A^nA'}$ on $A^{\otimes n}A'$ we have that
\bb
\left|S\left(\left(\MM_{A\to B}^{\otimes n}\otimes I_{A'}\right)(\rho_{A^nA'})\right) - S\left(\left(\NN_{A\to B}^{\otimes n}\otimes I_{A'}\right)(\rho_{A^nA'})\right)\right|\leq n\cdot\Big(\e \log \big( |B|^2 - 1\big) + h_2(\e)\Big)\,.
\ee
\end{lemma}

This improves on the previous bound~\cite[Theorem 11]{Leung2009}
\bb
\left|S\left(\left(\MM_{A\to B}^{\otimes n}\otimes I_{A'}\right)(\rho_{A^nA'})\right) - S\left(\left(\NN_{A\to B}^{\otimes n}\otimes I_{A'}\right)(\rho_{A^nA'})\right)\right|\leq n\cdot\Big(4\e \log|B| + 2h_2(\e)\Big)\,.
\ee

\begin{proof}[Proof of Lemma~\ref{lem:output-entropy}]
Following the telescoping argument from~\cite[Theorem 11]{Leung2009} and employing our improved continuity of the quantum conditional entropy for states with equal marginals from Theorem~\ref{equal_marginals_thm}, we immediately deduce the claim.
\end{proof}


\subsection{Improved continuity of entanglement cost} \label{sec:continuity_entanglement_cost}

The \deff{entanglement cost}~\cite{Horodecki-review} is the fundamental entanglement measure that quantifies the amount of pure entanglement that is needed to prepare a given mixed state in the asymptotic limit~\cite{Hayden-EC, Vidal-irreversibility, faithful-EC, EC-infinite} using only local operations and classical communication (LOCC)~\cite{LOCC}. It represents the maximal asymptotic entanglement measure~\cite{Horodecki2000, Donald2002, Horodecki2001}, and it can be computed for any state $\rho=\rho_{AB}$ of a bipartite quantum system $AB$ as
\begin{align}
E_c(\rho) =&\ \lim_{n\to\infty} \frac1n\, E_f\big(\rho^{\otimes n}\big)\, , \label{entanglement_cost} \\
E_f(\rho) \coloneqq&\ \inf_{\{p_x,\psi_x\}_x:\ \sum_x p_x \psi_x^{AB} \,=\, \rho_{AB}} \sum_x p_x\, S\big(\psi_x^A\big)\, . \label{entanglement_of_formation}
\end{align}
Throughout this section, we are going to assume that at least one between $A$ and $B$ is finite dimensional. The measure in~\eqref{entanglement_of_formation} is called the \deff{entanglement of formation}. The infimum in its definition runs over all decompositions of $\rho$ as a convex combination of pure states $\psi_x^{AB} = \ketbra{\psi_x}_{AB}$. For a precise operational definition of the entanglement cost, we refer the reader to~\cite{Hayden-EC} or to the explanation below~\eqref{rate}. However, for the sake of this discussion we could also take directly~\eqref{entanglement_cost} as a definition.

A continuity relation for the entanglement of formation in~\eqref{entanglement_of_formation} was proved already in 2000 by Nielsen~\cite{Nielsen2000} and then later refined in~\cite{Mark2020} based on the tight continuity bounds for the classical conditional entropy from~\cite{Alhejji2019}. The analysis of the continuity properties of the entanglement cost~\eqref{entanglement_cost} proved significantly more complicated, mainly due to the presence of the regularisation, i.e.\ the limit $n\to\infty$ over the number of copies, in its definition. The first continuity relation was established in 2015 by Winter~\cite[Corollary~4]{tightuniform}, and proving it required considerable technical progress~\cite[Proposition~5]{tightuniform}, which requires conditioning on quantum systems. Here, we show how to use our new Theorem~\ref{semicontinuity_relent_thm} to provide an improvement over Winter's result.

\begin{prop}
\label{prop:E-cost}
Let $AB$ be a bipartite quantum system, with either $A$ or $B$ finite dimensional. Let $\rho_{AB}$ and $\sigma_{AB}$ be two states on $AB$ whose trace distance satisfies
\bb
\frac12 \left\|\rho_{AB} - \sigma_{AB}\right\|_1 \leq \e \leq 1 - \frac{\sqrt{2d^2-1}}{d^2}\, .
\ee
Then, setting $\delta \coloneqq \sqrt{\e(2-\e)}$ and $d\coloneqq \min\{|A|,\,|B|\} <\infty$, we have that
\bb
\left| E_c(\rho_{AB}) - E_c(\sigma_{AB}) \right| \leq \delta \log\big(d^2 -1\big) + h_2(\delta)\, .
\ee
\end{prop}

This improves on the previous bound
\bb
\left| E_c(\rho_{AB}) - E_c(\sigma_{AB}) \right| \leq \delta \log\big(d^2\big) + (1+\delta)\cdot h_2\left(\frac{\delta}{1+\delta}\right)\
\ee
from~\cite[Corollary~4]{tightuniform}.

\begin{proof}[Proof of Proposition~\ref{prop:E-cost}]
The key observation is that Theorem~\ref{equal_marginals_thm} suffices to improve the continuity estimate in the original proof of the (asymptotic) continuity of the entanglement cost, given by Winter~\cite[Corollary~4]{tightuniform}. With Winter's notation, we can explain the idea very concisely as follows. The entire proof of~\cite[Corollary~4]{tightuniform} goes through, but in the last line of the last equation we apply our Theorem~\ref{equal_marginals_thm} instead of~\cite[Lemma~2]{tightuniform}. This is possible because $\widetilde{\rho}$ and $\widetilde{\sigma}$ have the same marginals on $A'X$ (in fact, even on $A'B'X$), as Winter observes in the equation before. Also, note that because of the restriction on the range of $\e$ we have that $\delta \leq 1 - \frac{1}{d^2}$.
\end{proof}

\begin{rem}
The technical reason why our fully quantum continuity relation is needed in the proof of continuity of the entanglement cost, but not in that of the entanglement of formation, is that in the aforementioned last equation of the proof of~\cite[Corollary~4]{tightuniform} the system $A'$ is un-measured and thus quantum. This does not happen in the proof of continuity for the entanglement of formation, as in that case the conditioning system is fully measured (cf.\ the first part of the proof of~\cite[Corollary~4]{tightuniform}).
\end{rem}


\subsection{Approximate degradability} \label{sec:approximate_degradability}

Let $\NN_{A\to B}$ be a quantum channel with input system $A$ and output system $B$. Let $V_{A\to BE}$ be one of its Stinespring dilations, so that $\NN_{A\to B}(\rho_A) = \Tr_E V_{A\to BE}^{\vphantom{\dag}} \rho_A V_{A\to BE}^\dag$ for all states $\rho_A$ on $A$. We denote with $\NN_{A\to E}^c$ the associated complementary channel, that acts as $\NN_{A\to E}^c(\rho_A) = \Tr_B V_{A\to BE}^{\vphantom{\dag}} \rho_A V_{A\to BE}^\dag$. For some $\e\in [0,1]$, we say that a channel $\Theta_{B\to \etilde}$, where $\etilde\simeq E$ is a copy of the quantum system $E$, is an \deff{$\boldsymbol{\e}$-degrading channel} for $\NN_{A\to B}$ if it holds that~\cite{Sutter-VV}
\bb
\frac12 \left\|\NN^c_{A\to E} - \iota_{\etilde \to E}\circ \Theta_{B\to \etilde} \circ \NN_{A\to B} \right\|_\diamond \leq \e\, ,
\ee
where $\iota_{\etilde\to E}$ is the embedding of $\etilde \simeq E$ into $E$, and the diamond norm is given by~\eqref{diamond_norm}.

The notion of approximate degradability was introduced to study quantum capacities~\cite{Leung2009,Sutter-VV,Leditzky18}. To compute these latter quantities, a useful notion is a suitably modified version of the coherent information. We remind the reader that the coherent information of a channel $\NN = \NN_{A\to B}$ with complementary channel $\NN^c = \NN^c_{A\to E}$ is given by
\bb
I_c \big(\NN\big) \coloneqq&\ \sup_{\rho = \rho_A} \left\{ S\big(\NN(\rho)\big) - S\big(\NN^c(\rho)\big) \right\}\,,
\label{coherent_information}
\ee
and due to the Lloyd--Shor--Devetak theorem~\cite{Lloyd-S-D, L-Shor-D, L-S-Devetak, temaconvariazioni}, the quantum capacity can be written as
\bb
Q\big(\NN\big) = \lim_{n\to\infty} \frac1n\, I_c\big( \NN^{\otimes n}\big) = \sup_{n\in \N^+} \frac1n\, I_c\big( \NN^{\otimes n}\big)\, .
\label{LSD}
\ee

For degradable (i.e.\ `$0$-approximate degradable') channels, the above regularisation is unnecessary, because the coherent information turns out to be additive; one thus obtains the single-letter formula $Q\big(\NN_{A\to B}\big) = I_c \big(\NN_{A\to B}\big)$. In the case of approximately degradable channels, the coherent information in general ceases to be additive and can be strictly super-additive, entailing that the regularisation in~\eqref{LSD} cannot be removed.

To obtain an additive quantity, one needs to consider a suitably modified version of the coherent information. Let $\Theta_{B\to \etilde}$ be an $\e$-approximate degrading map for $\NN_{A\to B}$. Then we can imagine to modify~\eqref{coherent_information} by replacing the complementary channel $\NN^c_{A\to E}$ on the right-hand side by its approximate version $\Theta_{B\to \etilde} \circ \NN_{A\to B}$ (we can safely ignore the isometric embedding $\iota_{\etilde\to E}$ if we only compute entropies). By doing so, one obtains the quantity~\cite[Eq.~(5)]{Sutter-VV}
\bb
U_\Theta(\NN) \coloneqq&\ \sup_{\rho = \rho_A} \left\{ S\big(\NN(\rho)\big) - S\big((\Theta\circ \NN)(\rho)\big) \right\} \\
=&\ \sup_{\rho_A} H(F|\tilde{E})\, ,
\label{U_Theta}
\ee
where $F$ is an environment for a Stinespring dilation $W_{B\to \etilde F}$ of $\Theta_{B\to F}$. We can compare directly~\eqref{U_Theta} and~\eqref{coherent_information} by re-writing the latter as
\bb
I_c(\NN) = \sup_{\rho_A} \left\{ H\big(\etilde F\big) - H(E) \right\} .
\label{coherent_information_systems}
\ee
Remarkably, the quantity in~\eqref{U_Theta} is additive, meaning that~\cite[Lemma~3.3]{Sutter-VV}
\bb
U_{\Theta_1\otimes \Theta_2}\big(\NN_1\otimes \NN_2\big) = U_{\Theta_1}(\NN_1) + U_{\Theta_2}(\NN_2)
\label{additivity_U_Theta}
\ee
holds for all choices of channels $\NN_1,\NN_2$ and $\Theta_1,\Theta_2$.

We are now ready to state the main result of this section.

\begin{prop} \label{bound_Q_prop}
Let $\NN$ be a quantum channel, and let $\Theta$ be an $\e$-degrading channel for $\NN$. Provided that $\e\leq 1 - \frac{1}{|E|^2}$, where $|E|$ is the dimension of the environment of a Stinespring dilation of $\NN$, the quantum capacity $Q(\NN)$ of $\NN$ can be bounded as
\bb
I_c(\NN) \leq Q(\NN) \leq U_\Theta(\NN) + \e \log\big( |E|^2 - 1 \big) + h_2(\e)\, ,
\label{bound_Q_with_U_Theta}
\ee
where $U_\Theta(\NN)$ is defined by~\eqref{U_Theta}. Furthermore, if $\e\leq 1-\frac{1}{|E|}$, then it also holds that
\bb
I_c(\NN) \leq Q(\NN) \leq I_c(\NN) + \e \log\big[ \big(|E| - 1)^2(|E|+1)\big] + 2h_2(\e)\, ,
\label{bound_Q_with_coherent_info}
\ee
\end{prop}

This improves on the previous bound in~\cite[Theorem~3.4]{Sutter-VV}, which featured worse fudge terms.

\begin{proof}[Proof of Proposition~\ref{bound_Q_prop}]
We follow the main ideas employed in the proof of~\cite[Theorem~3.4, (i)~and~(ii)]{Sutter-VV}, starting with~\eqref{bound_Q_with_U_Theta}. The first inequality is trivial, and follows by direct inspection of~\eqref{LSD}. As for the second one, pick a positive integer $n\in \N^+$, and write as in~\cite[proof of Theorem~3.4(ii)]{Sutter-VV}
\bb
H\big(\etilde^n F^n\big) - H(E^n) &= H\big(F^n|\etilde^n\big) + H\big(\etilde^n\big) - H(E^n) \\
&\eqt{(i)} H\big(F^n|\etilde^n\big) + \sum_{t=1}^n \left( H\big(\etilde_{\leq t} E_{>t}\big) - H\big(\etilde_{<t} E_{\geq t}\big) \right) \\
&\eqt{(ii)} H\big(F^n|\etilde^n\big) + \sum_{t=1}^n \left( H\big(\etilde_t \big| \etilde_{<t} E_{>t}\big) - H\big(E_t \big| \etilde_{<t} E_{>t}\big) \right) .
\label{bound_Q_proof_eq1}
\ee
Here, we used the shorthand notation $H\big(\etilde_{\leq t} E_{>t}\big) \coloneqq H\big(\etilde_1\ldots \etilde_t E_{t+1}\ldots E_n\big)$, and similarly $H\big(\etilde_{<t} E_{\geq t}\big) \coloneqq H\big(\etilde_1\ldots \etilde_{t-1}E_t\ldots E_n\big)$. We used a telescopic identity in~(i) and the chain rule in~(ii). To continue, we observe that the state of the system $\etilde_{\leq t} E_{>t}$ is of the form
\bb
\sigma^{(n,t)} = \sigma^{(n,t)}_{\etilde_{\leq t} E_{>t}} \coloneqq \left((\Theta\circ \NN)^{\otimes t} \otimes (\NN^c)^{\otimes (n-t)}\right)(\omega^n)\, ,
\ee
where $\omega^n = \omega^n_{A^n}$ is the initial state. The state of the system $\etilde_{<t} E_{\geq t}$ can instead be written as
\bb
\rho^{(n,t)} = \rho^{(n,t)}_{\etilde_{<t} E_{\geq t}} \coloneqq \left((\Theta\circ \NN)^{\otimes (t-1)} \otimes (\NN^c)^{\otimes (n-t+1)}\right)(\omega^n)\, .
\ee
Identifying $\etilde \simeq E$, we have that
\bb
\frac12 \left\|\rho^{(n,t)} - \sigma^{(n,t)}\right\|_1 \leq \frac12 \left\| \Theta\circ \NN - \NN^c \right\|_\diamond \leq \e \leq 1 - \frac{1}{|E|^2}\, ,
\ee
where $\etilde_{<t}E_{>t}$ plays the role of the ancillary system $A'$ in~\eqref{diamond_norm}. 
Note that
\bb
\rho^{(n,t)}_{\etilde_{<t} E_{>t}} = \Tr_{E_t} \rho^{(n,t)}_{\etilde_{<t} E_{\geq t}} = \left((\Theta\circ \NN)^{\otimes (t-1)} \otimes (\NN^c)^{\otimes (n-t)}\right)(\omega^n) = \Tr_{\etilde_t} \sigma^{(n,t)}_{\etilde_{\leq t} E_{>t}} = \sigma^{(n,t)}_{\etilde_{<t} E_{>t}}\, ,
\ee
i.e.\ the marginals of $\rho^{(n,t)}$ and $\sigma^{(n,t)}$ on $\etilde_{<t}E_{>t}$ coincide. 
This implies, via Theorem~\ref{equal_marginals_thm}, that
\bb
H\big(\etilde_t \big| \etilde_{<t} E_{>t}\big) - H\big(E_t \big| \etilde_{<t} E_{>t}\big) &= H\big(\etilde_t \big| \etilde_{<t} E_{>t}\big)_{\sigma^{(n,t)}} - H\big(E_t \big| \etilde_{<t} E_{>t}\big)_{\rho^{(n,t)}} \\
&\leq \e \log\big( |E|^2 - 1 \big) + h_2(\e)\, .
\label{bound_Q_proof_eq5}
\ee

Now, plugging~\eqref{bound_Q_proof_eq5} into~\eqref{bound_Q_proof_eq1}, optimising over the initial state $\omega^n$, and dividing by $n$ yields the inequality
\bb
\frac1n\, I_c\big(\NN^{\otimes n}\big) &\leq \frac1n\, U_{\Theta^{\otimes n}}\big(\NN^{\otimes n}\big) + \e \log\big( |E|^2 - 1 \big) + h_2(\e) \\
&\eqt{(iii)}\, U_\Theta(\NN) + \e \log\big( |E|^2 - 1 \big) + h_2(\e)\, , 
\label{bound_Q_proof_eq7}
\ee
where (iii)~follows from the additivity of $U_\Theta$ in~\eqref{additivity_U_Theta}. Taking the limit $n\to\infty$ in~\eqref{bound_Q_proof_eq7} yields~\eqref{bound_Q_with_U_Theta}, according to~\eqref{LSD}. 

As for~\eqref{bound_Q_with_coherent_info}, it is obtained precisely as in~\cite[Theorem~3.4(i)]{Sutter-VV}, by combining the above statement with the continuity relation that links $I_c(\NN)$ and $U_\Theta(\NN)$. This takes the form~\cite[Proposition~3.2]{Sutter-VV}
\bb
\left|I_c(\NN) - U_\Theta(\NN)\right| \leq \e \log(|E|-1) + h_2(\e)\, ,
\ee
provided that $\e\leq 1-\frac{1}{|E|}$.
\end{proof}

\begin{rem}
It is possible to obtain tighter versions of the relations~\eqref{bound_Q_with_U_Theta} and~\eqref{bound_Q_with_coherent_info} by using the full power of our main inequality~\eqref{semicontinuity_relent}. To this end, one needs to introduce two channel divergences, the \deff{un-stabilised} and the \deff{stabilised max-relative entropy}, respectively given by
\begin{align}
\widebar{D}_{\max}(\Lambda_1\|\Lambda_2) &\coloneqq \sup_\rho D_{\max}\big(\Lambda_1(\rho)\,\big\|\, \Lambda_2(\rho)\big)\, , \label{unstabilised_max_relent_channels} \\
D_{\max}(\Lambda_1\|\Lambda_2) &\coloneqq \sup_\rho D_{\max}\big((\Lambda_1\otimes I)(\rho)\,\big\|\, (\Lambda_2 \otimes I)(\rho)\big) \label{stabilised_max_relent_channels}
\end{align}
for any two quantum channels $\Lambda_1,\Lambda_2$ with the same input and output systems. Here, the max-relative entropy is defined by~\eqref{max_relative_entropy}, and the optimisation in~\eqref{unstabilised_max_relent_channels} is over all states on the input system, while in~\eqref{stabilised_max_relent_channels} we included also an ancillary system (a procedure known as `stabilisation'). Importantly, due to~\cite[Lemma~12]{Berta2018} the stabilised max-relative entropy can be computed via the semi-definite program
\bb
D_{\max}(\Lambda_1\|\Lambda_2) = \min\left\{ \gamma:\ 2^{\gamma} \Lambda_2 - \Lambda_1 \in \mathrm{CP}\right\} ,
\label{stabilised_max_relent_channels_Choi_states}
\ee
where $\mathrm{CP}$ denotes the set of completely positive maps. 

With this notation, we can state the following refinements of~\eqref{bound_Q_with_U_Theta} and~\eqref{bound_Q_with_coherent_info}: \emph{provided that $\e \leq 1 - 2^{- D_{\max}(\Theta \circ \NN\,\|\,\pi\Tr)}$, where $(\pi \Tr)(X) \coloneqq \pi_E \Tr X$ is the completely depolarising channel, we have that}
\bb
I_c(\NN) \leq Q(\NN) &\leq U_\Theta(\NN) + \e \log\left( 2^{D_{\max}(\Theta \circ \NN\,\|\,\pi\Tr)} - 1 \right) + h_2(\e)\, .
\ee
Moreover, \emph{if also the condition $\e \leq 1 - 1/M$ is satisfied for some $M\geq 2^{\widebar{D}_{\max}(\Theta \circ \NN\,\|\,\pi\Tr)}$ (for example, one could set $M = |E|$), then}
\bb
I_c(\NN) \leq Q(\NN) &\leq I_c(\NN) + \e \log\left[\left( M - 1 \right)\left( 2^{D_{\max}(\Theta \circ \NN\,\|\,\pi\Tr)} - 1 \right)\right] + 2h_2(\e)\, .
\ee
\end{rem}

All of these approximate degradability bounds can be evaluated both numerically and analytically on concrete quantum channels (cf.~\cite{Sutter-VV,Leditzky18}), and we note that our methods can be combined with other ways of extracting bounds (see, e.g., the latest~\cite{Jabbour23}).


\subsection{Improved continuity of filtered entropies} \label{sec:filtered_entropies}

Our goal in this section is to use the above Theorem~\ref{semicontinuity_relent_thm} to prove the following conjecture of Christandl, Ferrara, and Lancien~\cite[Conjecture~7]{random-private}, formulated below in a general fashion that the interested reader can find discussed in~\cite{random-private} below their statement of the problem.

\begin{cj}[{\cite{random-private}}] \label{random_private_cj}
Let $\FF\subseteq \D(\HH)$ be a set of states star-shaped around the maximally mixed state in dimension $D = \dim(\HH) < \infty$, and let $\LL\subseteq \cptp(\HH\to \HH')$ be any set of channels. Then, does there exist a universal constant $\kappa$ and a function $g:[0,1]\to \R_+$ with $\lim_{\e\to 0^+} g(\e) = 0$ such that for every $\rho,\sigma \in \D(\HH)$ satisfying
\bb
\frac12 \|\rho - \sigma\|_\LL \leq \e
\ee
we have that
\bb
\Big| D^\LL(\rho\|\FF) - D^\LL(\sigma\|\FF)\Big| \leqt{?} \kappa \e \log D +g(\e)\, .
\label{conjectured_continuity}
\ee
\end{cj}

Here, we defined
\begin{align}
\|X\|_\LL &\coloneqq \sup_{\Lambda\in \LL} \left\|\Lambda(X)\right\|_1\, , \label{L_norm} \\
D^\LL(\rho\|\omega) &\coloneqq \sup_{\Lambda\in \LL} D\big(\Lambda(\rho)\,\big\|\, \Lambda(\omega)\big) , \label{L_relent_two_states} \\
D^\LL(\rho\|\FF) &\coloneqq \inf_{\omega \in \FF} D^\LL(\rho\|\omega)\, . \label{L_relent_resource}
\end{align}
In essence,~\eqref{conjectured_continuity} is a sort of asymptotic continuity statement, but stronger, because it is given with respect to the `$\LL$-filtered' norm $\|\cdot\|_\LL$ instead of the full trace norm $\|\cdot\|_1$. If the full trace norm is considered, then asymptotic continuity is known, e.g.\ from~\cite[Theorem~11]{Schindler2023}.

The above conjecture is known to be true when $\LL$ is a set of measurements~\cite[Proposition~3]{rel-ent-sq}, or a specific subset of channels~\cite[Proposition~6]{random-private}. We will see that we can tackle the case of general sets of channels. Before we state the general result that will allow us to solve the conjecture, let us fix some terminology. An alternative measure of $\LL$-filtered distance from a set of states $\FF$ can be constructed by using the max-relative entropy~\eqref{max_relative_entropy} instead of the Umegaki relative entropy in~\eqref{L_relent_resource}. Namely,
\begin{align}
D^\LL_{\max}(\rho\|\omega) &\coloneqq \sup_{\Lambda\in \LL} D_{\max}\big(\Lambda(\rho)\,\big\|\, \Lambda(\omega)\big) , \label{L_max_relent_two_states} \\
D^\LL_{\max}(\rho\|\FF) &\coloneqq \inf_{\omega \in \FF} D^\LL_{\max}(\rho\|\omega)\, . \label{L_max_relent_resource}
\end{align}
Before proceeding, we also need to construct the real function
\bb
g(x)\coloneqq (1+x) \log(1+x) - x\log x\, ,
\label{g}
\ee
defined for all $x\geq 0$, with the usual convention that $0\log 0 = 0$. We now state the following general result.

\begin{prop} \label{filtered_semicontinuity_prop}
Let $\HH$ be an arbitrary (possibly infinite-dimensional) separable Hilbert space. Let $\FF\subseteq \D(\HH)$ be a set of states that is star-shaped around some $\tau\in \FF$, and let $\LL\subseteq \cptp(\HH\to \HH')$ be any set of channels. Then for all $\rho,\sigma \in \D(\HH)$ satisfying $D_{\max}(\rho\|\tau) < \infty$ and
\bb
\frac12 \|\rho - \sigma\|_\LL \leq \e
\ee
it holds that
\bb
D^\LL(\rho\|\FF) - D^\LL(\sigma\|\FF) &\leq \e D_{\max}^\LL(\rho\|\tau) + g(\e) + h_2(\e) \\
&\leq \e D_{\max}(\rho\|\tau) + g(\e) + h_2(\e)\, ,
\label{filtered_semicontinuity}
\ee
where the function $g$ is defined by~\eqref{g}.

In particular, if $\FF$ is convex, then minimising over $\tau\in \FF$ yields
\bb
D^\LL(\rho\|\FF) - D^\LL(\sigma\|\FF) \leq \e D_{\max}^\LL(\rho\|\FF) + g(\e) + h_2(\e)\, .
\label{filtered_semicontinuity_convex_set}
\ee
\end{prop}

As a simple corollary, we obtain a solution to Conjecture~\ref{random_private_cj} by Christandl, Ferrara, and Lancien.

\begin{cor} \label{filtered_asymptotic_continuity_cor}
Let $\FF\subseteq \D(\HH)$ be a set of states star-shaped around the maximally mixed state in dimension $D = \dim(\HH) < \infty$, and let $\LL\subseteq \cptp(\HH\to \HH')$ be any set of channels. Then for all $\rho,\sigma \in \D(\HH)$ satisfying
\bb
\frac12 \|\rho - \sigma\|_\LL \leq \e\,,
\ee
we have that
\bb
\Big| D^\LL(\rho\|\FF) - D^\LL(\sigma\|\FF)\Big| \leq \e \log D + h_2(\e) + g(\e)\, ,
\label{filtered_asymptotic_continuity}
\ee
where the function $g$ is defined by~\eqref{g}.
\end{cor}

The above result is remarkable because:
\begin{itemize}
\item it solves Conjecture~\ref{random_private_cj};
\item it generalises~\cite[Theorem~11]{Schindler2023} in three ways, namely (a)~lifting the asymptotic continuity from the trace distance to $\LL$-filtered distance (see also~\cite[Proposition~3]{rel-ent-sq}), (b)~removing the assumption on the convexity of $\FF$, and (c)~removing any assumption on the set $\LL$, which doesn't even need to be made of measurements;
\item it generalises and significantly tightens~\cite[Proposition~3]{rel-ent-sq} and~\cite[Proposition~6]{random-private}, removing any assumption on the set of channels, which in particular no longer need to be (partial) measurements; finally,
\item the estimate on the right-hand side seems reasonably tight, as its form is not too different from that found in~\cite[Lemma~7]{tightuniform} (note that $g(x) = (1+x)\, h_2\big(\frac{x}{1+x}\big)$).
\end{itemize}

We will now prove Proposition~\ref{filtered_semicontinuity_prop} and then deduce from it Corollary~\ref{filtered_asymptotic_continuity_cor}. In order to prove Proposition~\ref{filtered_semicontinuity_prop} with Theorem~\ref{semicontinuity_relent_thm}, we need to show that we can restrict the optimisation over free states in~\eqref{L_relent_resource} to states with respect to which $\rho$ has a controlled max-relative entropy (see~\eqref{max_relative_entropy} for a definition). This is needed because of the $D_{\max}$ term on the right-hand side of~\eqref{semicontinuity_relent}. We fix this with a little lemma that is already essentially contained in the proof of~\cite[Proposition~3]{rel-ent-sq} (and that we re-discovered independently).

\begin{lemma} \label{restricting_sigma_lemma}
For an arbitrary (possibly infinite-dimensional) separable Hilbert space $\HH$, let $\FF\subseteq \D(\HH)$ be a set of states star-shaped around some $\tau\in \FF$,\footnote{This means that for all $\sigma\in \FF$ and all $\lambda\in [0,1]$, it holds that $\lambda\sigma + (1-\lambda)\tau\in \FF$.} 
and let $\LL\subseteq \cptp(\HH\to \HH')$ be any set of channels. Then for all $\rho\in \D(\HH)$ and all $q\in (0,1)$, it holds that
\bb
\inf_{\omega\in \FF,\ \omega\, \geq\, q \tau} D^\LL(\rho\|\omega) + \log(1-q) \leq D^\LL(\rho\|\FF) \leq \inf_{\omega\in \FF,\ \omega\, \geq\, q \tau} D^\LL(\rho\|\omega)\, ,
\ee
where $D^\LL(\rho\|\FF)$ is defined by~\eqref{L_relent_resource}.
\end{lemma}

\begin{proof}
The upper bound is trivial. As for the lower bound, for every $\omega\in \FF$ one can construct $\omega' = (1-q) \omega + q \tau$, which is again a state in $\FF$ due to the fact that this set is star-shaped around $\tau$. Note also that $\omega\leq \frac{1}{1-q}\,\omega'$ and moreover $\omega' \geq q \tau$, so that by operator monotonicity of the logarithm (or, in other words, by applying~\cite[Proposition~5.1(i)]{PETZ-ENTROPY}) we obtain
\bb
D^\LL(\rho\|\omega) \geq D^\LL\Big(\rho \,\Big\|\, \frac{\omega'}{1-q} \Big) = D^\LL(\rho \| \omega' ) + \log(1-q) \geq \inf_{\omega''\in \FF,\ \omega''\geq q\tau} D^\LL(\rho \| \omega') + \log(1-q)\, .
\ee
Taking the infimum over $\omega\in \FF$ on the left-hand side completes the proof.
\end{proof}

We are now ready to prove Proposition~\ref{filtered_semicontinuity_prop} and Corollary~\ref{filtered_asymptotic_continuity_cor}.

\begin{proof}[Proof of Proposition~\ref{filtered_semicontinuity_prop}]
Let us fix an arbitrary $q\in (0,1)$. Then
\begin{align}
D^\LL(\rho\|\FF) &\leqt{(i)} \inf_{\omega\in \FF,\ \omega\, \geq\, q \tau} D^\LL(\rho\|\omega) \nonumber \\
&= \inf_{\omega\in \FF,\ \omega\, \geq\, q \tau} \sup_{\Lambda\in \LL} D\big(\Lambda(\rho)\,\big\|\, \Lambda(\omega)\big) \nonumber \\
&\leqt{(ii)} \inf_{\omega\in \FF,\ \omega\, \geq\, q \tau} \sup_{\Lambda\in \LL} \left\{ D\big(\Lambda(\rho)\,\big\|\, \Lambda(\omega)\big) + \e D_{\max}^\LL(\rho\|\tau) + \e \log \frac1q + h_2(\e) \right\} \nonumber \\
&\leqt{(iii)} D^\LL(\sigma\|\FF) + \log\frac{1}{1-q} + \e D_{\max}^\LL(\rho\|\tau) + \e \log \frac1q + h_2(\e)\, . \nonumber
\end{align}
Here: in~(i) we restricted the infimum (the easy direction in Lemma~\ref{restricting_sigma_lemma}); in~(ii) we applied Theorem~\ref{semicontinuity_relent_thm} in the simplified form~\eqref{simplified_semicontinuity_relent}, noting that $\frac 12 \big\|\Lambda(\rho) - \Lambda(\sigma)\big\|_1 \leq \frac12 \|\rho - \sigma\|_\LL\leq \e$ and observing that 
\bb
D_{\max}\big(\Lambda(\rho)\,\big\|\, \Lambda(\omega)\big) \leq D_{\max}\big(\Lambda(\rho)\,\big\|\, \Lambda(\tau)\big) + \log \frac1q \leq D_{\max}^\LL(\rho\|\tau) + \log \frac1q
\ee
via the triangle inequality for the max-relative entropy, applicable because $\omega\geq q\tau$; finally, in~(iii) we used again Lemma~\ref{restricting_sigma_lemma} (this time the less trivial part of it).

It remains only to optimise over $q\in (0,1)$. Fortunately this can be done analytically, thanks to the simple formula
\bb
g(\e) = \inf_{0<q<1} \left\{ \e\log \frac{1}{q} + \log\frac{1}{1-q}\right\} . 
\ee
Using it, we obtain that
\bb
D^\LL(\rho\|\FF) &\leq D^\LL(\sigma \|\FF) + \e D_{\max}^\LL(\rho\|\tau) + g(\e) + h_2(\e) \\
&\leqt{(iv)} D^\LL(\sigma \|\FF) + \e D_{\max}^\LL(\rho\|\tau) + g(\e) + h_2(\e)\, ,
\ee
where~(iv) is simply by data processing. Finally, if $\FF$ is convex, and hence star-shaped around all of its states, then we can minimise over $\tau\in \FF$, thus obtaining also~\eqref{filtered_semicontinuity_convex_set} and completing the proof.
\end{proof}

\begin{proof}[Proof of Corollary~\ref{filtered_asymptotic_continuity_cor}]
It suffices to apply Proposition~\ref{filtered_semicontinuity_prop} with $\tau = \id/D$ equal to the maximally mixed state. Indeed, noting that $\rho\leq \id = D \tau$, we have that $D_{\max}(\rho\|\FF) \leq D_{\max}(\rho\|\tau) \leq \log D$, and similarly for $\sigma$. Applying~\eqref{filtered_semicontinuity} twice, the second time with $\rho$ and $\sigma$ exchanged, we obtain precisely~\eqref{filtered_asymptotic_continuity}.
\end{proof}


\subsection{Upper bound on transformation rates in infinite dimensions} \label{sec:transformation_rates}

Let $\rho = \rho_{AB}$ and $\sigma = \sigma_{A'B'}$ be two bipartite states of possibly infinite-dimensional systems $AB$ and $A'B'$. A key quantity of interest in entanglement theory is the \deff{asymptotic transformation rate} $\rho\to \sigma$, i.e.\ the maximum rate of production of copies of $\sigma$ that can be achieved by consuming copies of $\rho$ and using only local operations and classical communication (LOCC)~\cite{LOCC}. This is formally defined by
\bb
R_\locc(\rho\to\sigma) \coloneqq \sup\left\{ R:\ \lim_{n\to\infty} \inf_{\Lambda\in \locc_n} \frac12 \left\| \Lambda_n\big(\rho^{\otimes n}\big) - \sigma^{\otimes \ceil{Rn}} \right\|_1 = 0 \right\} ,
\label{rate}
\ee
where $\locc_n$ denotes the set of channels with input system $A^n\!:\!B^n$ and output system ${A'}^m\!:\!{B'}^m$, and $m\coloneqq \ceil{Rn}$. The operationally most important entanglement measures, the distillable entanglement and the entanglement cost (see also Section~\ref{sec:continuity_entanglement_cost}), can be seen as special case of the above notion, according to the formulae $E_{d,\,\locc}(\rho) \coloneqq R_\locc(\rho\to\Phi_2)$ and $E_{c,\,\locc}(\rho) \coloneqq R_\locc(\Phi_2\to \rho)^{-1}$. 

In spite of their operational importance, asymptotic transformation rates are extremely difficult to evaluate and especially estimate from above. For this reason, one is often interested in tight upper bounds on them. A general strategy to construct upper bounds is to use the \deff{regularised relative entropy of entanglement}, defined by~\cite{Vedral1997, Vedral1998}
\begin{align}
E_R^\infty(\rho) &\coloneqq \lim_{n\to\infty} \frac1n\, E_R\big(\rho^{\otimes n}\big)\, , \label{regularised_relent_entanglement} \\
E_R(\rho) &\coloneqq D\big(\rho\,\big\|\, \SEP(A\!:\!B)\big) = \min_{\omega\in \SEP(A:B)} D(\rho\|\omega)\, . \label{relent_entanglement}
\end{align}
Here, $\SEP(A\!:\!B)$ denotes the set of \deff{separable states} on the bipartite system $AB$, defined by~\cite{Werner}
\bb
\SEP(A\!:\!B) \coloneqq \cl\co\left\{\ketbra{\alpha}_A\otimes \ketbra{\beta}_B :\ \ket{\alpha}_A \in \HH_A,\ \ket{\beta}_B\in \HH_B,\ \braket{\alpha|\alpha} = 1 = \braket{\beta|\beta}\right\} ,
\ee
where $\cl \co$ represents the closed convex hull (with respect to the trace norm topology). The limit in~\eqref{regularised_relent_entanglement} exists due to Fekete's lemma~\cite{Fekete1923}, because of the sub-additivity property $E_R\big(\rho^{\otimes n}\big) \leq n E_R(\rho)$. Furthermore, the infimum in the definition of $D\big(\rho\,\big\|\, \SEP(A\!:\!B)\big)$ in~\eqref{relent_entanglement} is always achieved, also for infinite-dimensional systems~\cite[Theorem~5]{achievability}. The importance of the regularised relative entropy of entanglement stems from its operational interpretation in the context of entanglement testing, guaranteed by the recently established generalised quantum Stein's lemma~\cite{Brandao2010, gap, Hayashi-Stein, GQSL}.

If $\rho$ and $\sigma$ are finite-dimensional states, with $\sigma \notin \SEP(A\!:\!B)$ entangled, i.e.\ not separable, it is known that~\cite[Eq.~(146)]{Horodecki-review}
\bb
R_\locc(\rho\to\sigma) \leq \frac{E_R^\infty(\rho)}{E_R^\infty(\sigma)}\, .
\label{bound_rates_finite_dim}
\ee
Note that $E_R^\infty(\sigma)>0$, since $E_R^\infty$ is a faithful measure of entanglement~\cite{Brandao2010, Piani2009}. The problem with~\eqref{bound_rates_finite_dim} is that its standard proof~\cite[XV.E.2]{Horodecki-review} makes use of \emph{asymptotic continuity}, a strong form of continuity that does not apply to infinite-dimensional systems. Our main contribution in this section is to generalise~\eqref{bound_rates_finite_dim} to the infinite-dimensional setting, thus avoiding the `asymptotic continuity catastrophe' described in~\cite{nonclassicality}. The price we will have to pay is a stronger regularity assumption on the target state $\sigma$; however, crucially, no regularity assumption is made on the LOCC protocols that can be employed in~\eqref{rate} to effect the transformation. Our result is as follows.

\begin{prop} \label{bound_rates_general_prop}
Let $AB$, $A'B'$ be arbitrary bipartite quantum systems, finite or infinite dimensional. Let $\rho=\rho_{AB}$ and $\sigma=\sigma_{A'B'}$ be two states with $\sigma \notin \SEP(A'\!:\!B')$ and $E_{R,\,\max}(\sigma) < \infty$. Then
\bb
R_{\locc}(\rho\to\sigma) \leq R_{\sepp}(\rho\to\sigma) \leq \frac{E_R^\infty(\rho)}{E_R^\infty(\sigma)}\, .
\label{bound_rates_general}
\ee
\end{prop}

In the above statement, $E_{R,\,\max}$ denotes the max-relative entropy version of~\eqref{relent_entanglement}, i.e.
\bb
E_{R,\,\max}(\rho) \coloneqq \min_{\omega\in \SEP(A:B)} D_{\max}(\rho\|\omega)\, ,
\ee
also called the \deff{logarithmic robustness of entanglement}~\cite{VidalTarrach, Steiner2003, Datta-alias}. As for the standard relative entropy of entanglement, also in this case we have the sub-additivity property
\bb
E_{R,\,\max}\big(\rho^{\otimes n}\big) \leq n\, E_{R,\,\max}(\rho)\, .
\label{subadditivity_E_R_max}
\ee
Furthermore, in~\eqref{bound_rates_general} we denoted with $R_{\sepp}(\rho\to\sigma)$ the asymptotic transformation rate under \deff{non-entangling (NE) operations}~\cite{BrandaoPlenio1, BrandaoPlenio2, gap, irreversibility}, obtained from~\eqref{rate} by replacing LOCC with the set of all channels that map separable states to separable states.

\begin{proof}[Proof of Proposition~\ref{bound_rates_general_prop}]
Since the set of LOCC operations is a subset of that of NE operations, it is immediate to see that
\bb
R_{\locc}(\rho\to \sigma) \leq R_{\sepp}(\rho\to \sigma)\, .
\ee
Thus, it suffices to prove the second inequality in~\eqref{bound_rates_general}. To this end, let $R$ be an achievable rate, i.e.\ let it belong to the set in~\eqref{rate}, with LOCC replaced by NE. Then there exists a sequence of NE operations $\Lambda_n$ such that the states $\omega_n \coloneqq \Lambda_n\big(\rho^{\otimes n}\big)$ satisfy $\e_n \coloneqq \frac12 \left\| \omega_n - \sigma^{\otimes \ceil{Rn}} \right\|_1 \ctends{}{n\to\infty}{1.5pt} 0$. We can now write that
\bb
E_R\big(\sigma^{\otimes \ceil{Rn}}\big) - E_R\big(\rho^{\otimes n}\big) &\leqt{(i)} E_R\big(\sigma^{\otimes \ceil{Rn}}\big) - E_R(\omega_n) \\
&\leqt{(ii)} \e_n\, E_{R,\,\max}\big(\sigma^{\otimes \ceil{Rn}}\big) + g(\e_n) + h_2(\e_n) \\
&\leqt{(iii)} \ceil{Rn} \e_n\, E_{R,\,\max}(\sigma) + g(\e_n) + h_2(\e_n)\, .
\label{bound_rates_general_proof_key_inequalities}
\ee
The justification of the above inequalities is as follows. In~(i) we simply observed that due to monotonicity under NE operations~\cite[Lemma~IV.1]{BrandaoPlenio2} it holds that $E_R\big(\rho^{\otimes n}\big) \geq E_R\big(\Lambda_n\big(\rho^{\otimes n}\big)\big) = E_R(\omega_n)$. In~(ii) we applied Proposition~\ref{filtered_semicontinuity_prop}, and in particular~\eqref{filtered_semicontinuity_convex_set}, with $\FF = \SEP\big({A'}^{\ceil{Rn}}\!:\!{B'}^{\ceil{Rn}}\big)$ and $\LL = \cptp\big(\HH_{AB} \to \HH_{AB}\big)$, so that $D^\LL(\cdot \| \FF) = E_R(\cdot)$ and $D^\LL_{\max} (\cdot \| \FF) = E_{R,\,\max}(\cdot)$. Finally, in~(iii) we employed~\eqref{subadditivity_E_R_max}.

Diving both sides of~\eqref{bound_rates_general_proof_key_inequalities} by $n$ and taking the limit $n\to\infty$ yields immediately $R E_R^\infty(\sigma) - E_R^\infty(\rho) \leq 0$. The claim follows by taking the supremum over $R$ and dividing by $E_R^\infty(\sigma)$, which is strictly positive due to the faithfulness of $E_R^\infty$, whose proof in~\cite{Piani2009} works equally well for finite- and infinite-dimensional systems.
\end{proof}


\section{Conclusion}

Further applications in quantum information are waiting to be explored, but the main open question remains Conjecture~\ref{cj:wilde-conjecture} about the conditional entropy for general bipartite quantum states. Whereas many of applications do seem to have fixed marginals, one important question we have to leave open concerns tight continuity bounds for quantum conditional mutual information, which would improve on~\cite{Alicki-Fannes, Synak2006, mosonyi11, Winter2016}. A strictly related question is that of establishing tight continuity bounds for the squashed entanglement~\cite{Tucci1999, squashed}, improving on the original one by Alicki--Fannes~\cite{Alicki-Fannes}.

Another special case that is particularly relevant for improving the known continuity bounds on capacities of quantum channels (see~\cite{Leung2009,Shirokov17}) would be a version of Theorem~\ref{equal_marginals_thm} in terms of quantum mutual information, where one marginal of the bipartite states is equal (but not the other one). This is not covered by Theorem~\ref{semicontinuity_relent_thm} and our proof techniques, but one might nevertheless conjecture even more generally that for states $\rho_{AB},\sigma_{AB}$ and with trace distance $\frac12 \|\rho_{AB} - \sigma_{AB}\|_1 \leq \e$ small enough, one has something like
\bb
\label{eq:mutual-information-conj}
\left| I(A:B)_\rho - I(A:B)_\sigma \right| \leqt{?} \e \log \Big( \min\left\{|A|^2,|B|^2\right\} - 1\Big) + h_2(\e)
\ee
with the quantum mutual information $I(A:B)_\rho\coloneqq H(A)_\rho+H(B)_\rho-H(AB)_\rho$. In fact, to the best of our knowledge, this question seems to be open even in the classical case.\footnote{We have not intensively tested the above mutual information inequality numerically.}

\begin{note}
Our result on the uniform continuity bound for the conditional entropy in the case of equal marginals (Theorem~\ref{equal_marginals_thm}) was obtained with different techniques by Audenaert, Bergh, Datta, Jabbour, Capel, and Gondolf.
\end{note}


\paragraph*{Acknowledgments:}

MB acknowledges funding by the European Research Council (ERC Grant Agreement No.\ 948139) and support from the Excellence Cluster -- Matter and Light for Quantum Computing (ML4Q). MT is supported by the National Research Foundation, Singapore and A*STAR under its CQT Bridging Grant. We thank Maksim E.\ Shirokov for feedback on a previous version of our work and for pointing us to his work~\cite{Shirokov17}. LL and MB thank Nilanjana Datta and Micheal Jabbour for discussions on our preliminary results in Cambridge and Budapest. We thank Mohammad A.\ Alhejji, Gereon Kossmann, and Roberto Rubboli for various discussions about approaches to extend Theorem~\ref{equal_marginals_thm} to general quantum states (as stated in Conjecture~\ref{cj:wilde-conjecture} taken from~\cite{Mark2020}). LL is also grateful to Francesco Anna Mele for enlightening discussions on approximate degradability.

\bibliography{biblio}

\end{document}